\documentclass[]{article}
\usepackage[a4paper, margin=3cm]{geometry}

\usepackage[UKenglish]{babel} 
\usepackage[utf8]{inputenc}
\usepackage{epsfig,graphicx,color,xcolor,float}
\usepackage{latexsym, gensymb} 
\usepackage{amsmath,amssymb,amsfonts,verbatim,xkeyval}
\usepackage{bm, amsthm, amscd,wasysym,nccmath, centernot}
\usepackage{enumerate}
\usepackage{listings}
\usepackage{dsfont}
\usepackage{multicol}
\usepackage{comment}
\usepackage{physics}
\usepackage{mathtools}
\usepackage{tikz,tikz-3dplot}
\usepackage{scalerel}
\usepackage{cancel}
\usepackage{listings}
\usepackage{setspace}
\usepackage{mathdots}
\usepackage[normalem]{ulem}

\usepackage{microtype}

\usepackage[autostyle]{csquotes}
\usepackage[
backend=biber,
style=alphabetic,
sorting=nyt,
hyperref=true,
backref=false,
backrefstyle=none
]{biblatex}
\usepackage{hyperref}
\hypersetup{
	colorlinks=true,
	linkcolor=black,
	filecolor=black,
	citecolor = black,      
	urlcolor=black,
}

\usepackage{cleveref}

\usepackage{authblk}
\author[1a]{Adriano Chialastri}
\author[2b]{Ko Sanders}
\affil[1]{SISSA, Mathematics Area, Via Bonomea 265, 34136 Trieste, Italy}
 \affil[a]{\href{mailto:achialas@sissa.it}{achialas@sissa.it}}
\affil[2]{Leibniz Universität Hannover, Institut für Analysis, Welfengarten 1,
30167 Hannover, Germany}
\affil[b]{\href{mailto:ko.sanders@math.uni-hannover.de}{ko.sanders@math.uni-hannover.de}}

\newtheorem{theorem}{Theorem}[section]
\newtheorem{corollary}[theorem]{Corollary}
\newtheorem{proposition}[theorem]{Proposition}
\newtheorem{lemma}[theorem]{Lemma}

\newtheorem{remark}[theorem]{Remark}

\def\Rbb{\mathbb{R}}
\def\Cbb{\mathbb{C}}
\def\H{\mathcal{H}}
\def\K{\mathcal{K}}
\def\N{\mathcal{N}}
\def\O{\mathcal{O}}
\def\one{\mathds{1}}
\newcommand{\cstar}{\(C^*\)}

\def\theta{\vartheta}
\def\phi{\varphi}
\def\epsilon{\varepsilon}
\def\artanh{\operatorname{artanh}}

\title{\textbf{Quasi-equivalence of Gaussian states and energy estimates for functions of modular Hamiltonians}}
\date{14 October 2025}
\numberwithin{equation}{section}

\begin{document}
	
\maketitle


\begin{abstract}
To compare two Gaussian states of the Weyl-CCR algebra of a free scalar QFT we study three closely related perspectives: (i) quasi-equivalence of the GNS-representations, (ii) differences of the total energy (on some Cauchy surface), and (iii) differences between functions of the modular Hamiltonians. (For perspective (ii) we will only consider real linear free scalar quantum fields on ultrastatic spacetimes.) These three perspectives are known to be related qualitatively, due to work of Araki and Yamagami, Verch and Longo. Our aim is to investigate quantitative relations, including in particular estimates of differences between functions of modular Hamiltonians in terms of energy differences.
E.g., for a suitable class of perturbations of the Minkowski vacuum state of a  
massive free scalar field, which have a positive energy density and a finite total energy \(E\) on some inertial time slice, the modular Hamiltonian \(K\) satisfies \(\Big\|\frac{1}{\cosh\left(\frac{K}{2}\right)}\Big\|_{\text{HS}}^2\le 8\frac{E}{m}\).
\end{abstract}

\section{Introduction}\label{section:Introduction}

\cstar-algebras often admit many representations. In order to express the idea that the choice of representation does not matter, one may consider representations that are pairwise quasi-equivalent. This condition, which is weaker than unitary equivalence, ensures that every state that can be described by a density operator in one representation can also be represented by a density operator in another representation. E.g. in quantum field theories, especially in curved spacetime, a theory typically admits many representations that are not unitarily equivalent, but that are, at least locally, quasi-equivalent, so that the choice of a Hilbert space representation is irrelevant for describing local physics \cite{witten2022does,verch1994local}.

For a bosonic quantum system satisfying canonical commutation relations (CCR), the earliest investigation into the equivalence of Gaussian states of the Weyl algebra appears to be Shale's work on unitary equivalence of pure states~\cite{Shale1962}. In the most general situation, allowing (degenerate) pre-symplectic forms and Gaussian states that need not be pure, necessary and sufficient conditions for quasi-equivalence were established by Araki and Yamagami~\cite{araki1982quasi}. These general conditions were verified for the class of Gaussian Hadamard states of a linear real scalar quantum field on bounded regions of any globally hyperbolic Lorentzian manifold by Verch~\cite{verch1994local}. It is gratifying to see that the Hadamard condition ensures local quasi-equivalence, because this condition has been identified as physically important and these states have finite and smoothly varying expectation values for the renormalised stress tensor. This class includes ground and thermal (KMS) states on stationary spacetimes \cite{SahlmannVerch2000} and is preserved under scattering against smooth perturbations in the metric \cite{FSW1978}. Because a smooth energy density on a bounded region integrates to a finite total energy, one expects that a finite change in energy is insufficient to violate quasi-equivalence.

We will review the conditions of Araki-Yamagami in Section~\ref{sec:perturbations}, together with a recent development by Longo~\cite{longo2022modular} (see also \cite{conti2024quasi}), who found equivalent conditions in a slightly more limited setting, formulated first in terms of polarisation operators and then in terms of modular Hamiltonians. The connection to modular theory is of great interest, because it is closely related to concepts from quantum information theory, as generalised to quantum field theory. In particular, the relative entropy between two states can be expressed in terms of relative modular operators, using results of Araki~\cite{Araki1975,Araki1977}.

The relation between quasi-equivalence and modular theory raises the suggestion that relative entropies (or other quantities that can naturally be formulated in terms of modular theory) are closely connected to quasi-equivalence of states~\cite{longo2025bekenstein,hollands2025bekenstein}.
Hence, for free scalar fields, they should also be closely related to energy differences. 
This relation between relative entropies and energy differences is
similar in nature to the Bekenstein bound \cite{PhysRevD.23.287}.

The purpose of this paper is to establish quantitative relations between the works of Araki and Yamagami~\cite{araki1982quasi}, Verch~\cite{verch1994local} and Longo~\cite{longo2022modular} that connect the perspectives of quasi-equivalence, finite energy differences and the comparison of modular operators. In particular, we supplement the existing results by explicit estimates, which allow us to estimate the difference of suitable functions of the modular Hamiltonians in two Gaussian states in terms of the difference of their energies. (For a rather different approach to estimates between modular operators and energy see \cite{MuchPasseggerVerch2022}.) The estimates that we present rely on inequalites by Powers and St{\o}rmer~\cite{powers1970free}, van Hemmen and Ando~\cite{van1980inequality}, Kittaneh \cite{kittaneh1987inequalitiesII} and Kosaki \cite{kosaki1998arithmetic}.

To avoid complications in our analysis we will focus on easy applications, rather than full generality.

We have organised our paper as follows. In Section~\ref{section:notations} we set the stage by introducing basic notations and reviewing key constructions relating to the Weyl algebra, Gaussian states, standard subspaces and modular theory. In Section~\ref{sec:perturbations} we review the quasi-equivalence results of Araki-Yamagami~\cite{araki1982quasi} and of Longo~\cite{longo2022modular} and we present our explicit estimates. Applications to the linear real scalar quantum field will be presented in Section \ref{sec:scalarfields}, establishing inequalities between differences of functions of modular operators and differences in energy. This also establishes a relation to the work of Verch~\cite{verch1994local} on quasi-equivalence of Hadamard states. 

\section{Notations and basic constructions}\label{section:notations}

In this section we will introduce some notations and general constructions regarding (pre-)symplectic spaces, one-particle structures, standard subspaces and modular theory, that will be used throughout the paper. By a \emph{pre-symplectic space} \((D,\sigma)\) we mean a real vector space \(D\) with an anti-symmetric bilinear form \(\sigma\) defined on it. We speak of a \emph{symplectic space} when \(\sigma\) is non-degenerate, i.e. when \(\sigma(f,g)=0\) for all \(g\in D\) implies \(f=0\).

To \((D,\sigma)\) we can associate a\emph{Weyl algebra} \(\mathcal{A}\left[ D,\sigma\right] \), generated by linearly independent unitaries \(W\left( f\right) \), \(f\in D\), satisfying \(W(f)^*=W(-f)\) and the Weyl relations
\begin{align}
	W\left( f\right) W\left( g\right)&=e^{-\frac{i}{2}\sigma\left(f,g\right) } W\left( f+g\right)\label{eqn:Weylrelations}
\end{align}
for all \(f,g\in D\). A \emph{Gaussian state} on \(\mathcal{A}\left[ D,\sigma\right] \) is a linear functional \(\omega:\mathcal{A}\left[ D,\sigma\right]\to\Cbb\) such that
\begin{align}
\omega\left( W\left( f\right) \right) &=\exp\left( -\frac{1}{4}\mu\left( f,f\right) \right)\,,\label{eqn:Gaussianstate}
\end{align}
where \(\mu\) is a real inner product on \(D\) such that the two-point distribution
\begin{align}
\omega_2(f,g)&:=\frac12(\mu(f,g)+i\sigma(f,g))\,,\label{def:omega2}
\end{align}
extended linearly to the standard complexification of \(D\), 
defines a Hermitean inner product by \(\omega_2(\bar{f},g)\). Gaussian states can be characterised entirely in terms of their one-particle structure \cite{Kay1978}.

A \emph{one-particle structure} for a pre-symplectic space \((D,\sigma)\) is a real-linear map \(\kappa:D\to\H\) into a complex Hilbert space \(\H\) with inner product \( (,) \) such that
(i) \(\kappa(D)+i\kappa(D)\) is dense in \(\H\) and (ii) \(\Im(\kappa(f),\kappa(g))=\sigma(f,g)\) for all \(f,g\in D\). Given a one-particle structure we can define an inner product \(\mu\) on \(D\) by setting
\begin{align}
\mu(f,g)&:=\Re(\kappa(f),\kappa(g))\,.\label{def:mu}     
\end{align}
We note that \(\mu\) and \((,)\) determine the same norm \(\|.\|\) on \(D\) and that
\begin{align}
|\sigma(f,g)|&\le\|f\|\cdot\|g\|\label{eqn:mudominates}
\end{align}
for all \(f,g\in D\), because of
\(|\sigma(f,g)|\le|(\kappa(f),\kappa(g))|\) and the Cauchy-Schwarz inequality.

Conversely, given any inner product \(\mu\) on \(D\) satisfying \eqref{eqn:mudominates} there exists a one-particle structure \(\kappa\) such that \eqref{def:mu} holds. This one-particle structure is unique up to unitary equivalence. We will now briefly recall the construction of \(\kappa\), because it allows us to introduce some further useful notations.

By a \emph{polarisation operator} \(R\) on a real Hilbert space \(\K_{\Rbb}\) with inner product \(\langle,\rangle\) we will mean an anti-selfadjoint operator, \(R^*=-R\), with \(\|R\|\le1\). In the case of a one-particle structure we obtain such a polarisation operator as follows. We can complete \(D\) in the inner product \(\mu\) to a real Hilbert space, which we will denote by \(\K_{\Rbb,\mu}\). On this Hilbert space there is a unique, bounded operator \(R_{\mu}\) defined by
\begin{align}
\mu(f,R_{\mu}g)&:=\sigma(f,g)\label{eqn:defR}
\end{align}
for all \(f,g\in D\). As the notation suggests, \(R_{\mu}\) is a polarisation operator.

For any real Hilbert space \(\K_{\Rbb}\) there is a standard complexification \(\K=\K_{\Rbb}\oplus i\K_{\Rbb}\) with multiplication by \(i\) given by the matrix \(\begin{pmatrix}0&-\one\\ \one&0\end{pmatrix}\), where the real inner product is extended to a Hermitian inner product. \(\K\) also carries a canonical complex conjugation \(\Gamma\) such that \(\Gamma(f+ig)=f-ig\) when \(f,g\in\K_{\Rbb}\). Any real-linear operator \(A\) on \(\K_{\Rbb}\) can be extended to a complex-linear operator on \(\K\) that we will denote by the same symbol. We then have \(\Gamma A\Gamma=A\) and we note that the extension does not increase the operator norm. In particular, if \(R\) is a polarisation operator on \(\K_{\Rbb}\), then its extension is also
anti-selfadjoint with \(\|R\|\le1\) on \(\K\) and we will write \(\Sigma:=iR\), which is a self-adjoint operator with \(\|\Sigma\|\le1\) and \(\Gamma\Sigma\Gamma=-\Sigma\). (When the real Hilbert space and the polarisation operator are derived from an inner product \(\mu\) on a symplectic space \((D,\sigma)\) satisfying \eqref{eqn:mudominates}, we will put a subscript \(\mu\) on the Hilbert spaces and the operators.)

Given a polarisation operator \(R_{\mu}\) on a real Hilbert space \(\K_{\Rbb,\mu}\) we let \(\N_{\mu}:=\mathrm{ker}(\one+\Sigma_{\mu})\) and \(\H_{\mu}:=\N_{\mu}^{\perp}\) in \(\K_{\mu}\). Furthermore, we define the real-linear map \(\kappa:D\to\H_{\mu}\) by
\begin{align}
\kappa(f)&:=\sqrt{\one+\Sigma_{\mu}}f\,.\label{def:kappa}
\end{align}
Then, \(\kappa\) is a one-particle structure for \((D,\sigma)\) satisfying \eqref{def:mu}.

\medskip

We now turn our attention to modular theory. Our presentation essentially follows that of~\cite[Sec. 2]{longo2008real}. A \emph{standard subspace} of a complex Hilbert space \(\H\) is a closed, real-linear subspace \(H\subset\H\) such that \(H+iH\subset\H\) is dense and \(H\cap iH=\{0\}\). Given a standard subspace we can define the \emph{Tomita operator}
\begin{align}
S_H:\ &H+iH\to H+iH\label{def:Tomita}\\
&\ v+iw\mapsto v-iw\notag
\end{align}
which is well-defined on a dense domain and closed. It has a polar decomposition \(S_H=J_H\Delta_H^{\frac12}\), where \(\Delta_H\ge0\) is the \emph{modular operator} and \(J_H\), the \emph{modular conjugation}, is anti-linear. From \(S_H^2=\one\) we see that \(\Delta_H^{\frac12}\) is invertible with \(J_H\Delta_H^{\frac12}J_H=\Delta_H^{-\frac12}\). In particular, \(\Delta_H>0\) is strictly positive. Also, \(J_H\) is invertible with \(J_H^{-1}=J_H^*\) and hence \(S_H=\Delta_H^{-\frac12}J_H^*\) and \(S_H^*=J_H\Delta_H^{-\frac12}\). From \(\Delta_H=S_H^*S_H=J_H\Delta_H^{-1}J_H^*\) and spectral calculus one then finds \(J_H\Delta_H^{\frac12}J_H^*=\Delta_H^{-\frac12}=J_H\Delta_H^{\frac12}J_H\), which 
implies \(J_H^*=J_H\).

The \emph{modular group} \(t\mapsto \Delta_H^{-it}\) is defined in terms of the modular operator and generated by the \emph{modular Hamiltonian}
\begin{equation}
K_H:=-\log(\Delta_H)\,.\label{def:modularHamiltonian}
\end{equation}
We call a standard subspace \(H\subset\H\) \emph{factorial} iff 1 is not in the point spectrum of \(\Delta_H\). Equivalently, \(H\cap H'=\{0\}\), where
\begin{equation}
H':=\{v\in\H\mid \Im(v,w)=0\,  ,\ \forall w\in H\}    
\end{equation}
is the \emph{symplectic complement}.

\medskip

We can relate modular theory to polarisation operators as follows. Given a standard subspace \(H\subset\H\) we can view \(\K_{\Rbb}=H\) as a real Hilbert space with inner product \(\langle,\rangle=\Re(,)\). Furthermore, we can define a polarisation operator \(R\) on \(H\) by \(\langle v,Rw\rangle=\Im(v,w)\) for all \(v,w\in \K_{\Rbb}\). We let \(\K\) denote the standard complexification of \(\K_{\Rbb}\) as before and we note that for all \(v,w\in H\) we have \(\|v+iw\|^2_{\H}=\|\sqrt{\one+\Sigma}\,(v+iw)\|^2_{\K}\) by a short computation. Because \(H\) is a standard subspace, both sides must be non-zero, unless \(v=w=0\). This means that \(\one+\Sigma\) is strictly positive and consequently \(\N=\{0\}\), so we can identify \(\H\simeq\K\) using the unitary map
\begin{align}
U:\ &\H\to\K: v+iw\mapsto\sqrt{\one+\Sigma}\,(v+iw)\label{eqn:defU}
\end{align}
for all \(v,w\in H\). Note that \(-1\) is not in the point spectrum of \(\Sigma\) and neither is \(+1\), because \(\Sigma=-\Gamma\Sigma\Gamma\). Equivalently we have
\begin{align}
\one+R^2&=\one-R^*R>0\,.\label{Radditionalcondition}
\end{align}

We can express the modular operator in terms of \(\Sigma\) using the following computation for all \(v,w\in H\):
\begin{align}
\|\Delta_H^{\frac12}(v+iw)\|_{\H}&=\|S_H(v+iw)\|_{\H}\notag\\
&=\|v-iw\|_{\H}\notag\\
&=\|\sqrt{\one+\Sigma}(v-iw)\|_{\K}\notag\\
&=\|\Gamma\sqrt{\one+\Sigma}\Gamma(v+iw)\|_{\K}\notag\\
&=\|\sqrt{\one-\Sigma}(v+iw)\|_{\K}\notag\\
&=\left\|\sqrt{\frac{\one-\Sigma}{\one+\Sigma}}U(v+iw)\right\|_{\K}\notag\\
&=\left\|U^*\sqrt{\frac{\one-\Sigma}{\one+\Sigma}}U(v+iw)\right\|_{\H}\notag
\end{align}
which implies \(\Delta_H^{\frac12}=U^*\sqrt{\frac{\one-\Sigma}{\one+\Sigma}}U\) by the uniqueness of the polar decomposition of \(\Delta_H^{\frac12}\). Hence we find (see also \cite[Proposition 2.4]{longo2022modular} or \cite[Proposition 2.1]{conti2024quasi})
\begin{equation}\label{deltaRidentification}
\begin{aligned}
U\Delta_HU^*&=\frac{\one-\Sigma}{\one+\Sigma}=\frac{\one-iR}{\one+iR}\,,\\
U^*\Sigma U&=
-i\,U^*RU=
\frac{\one-\Delta_H}{\one+\Delta_H}\,.
\end{aligned}
\end{equation}

\medskip

When \(\kappa:D\to\H\) is a one-particle structure for a pre-symplectic space \((D,\sigma)\), then the closed range \(H:=\overline{\kappa(D)}\) has a dense complexification \(H+iH\) in \(\H\). \(H\) is a standard subspace iff \(R_{\mu}\) satisfies \eqref{Radditionalcondition}. In this case, we will use a subscript \(\mu\) also on the modular operator.

\section{Perturbed states and quasi-equivalence}\label{sec:perturbations}

On the pre-symplectic space \((D,\sigma)\) we now fix a Gaussian reference state, indicated by the real inner product \(\mu_0\) on \(D\). We will consider perturbations \(\mu_{\delta}\) of \(\mu_0\) and consider the quasi-equivalence of the corresponding Gaussian states, following~\cite{araki1982quasi} and~\cite{longo2022modular}. We supplement these investigations with some explicit estimates that will be useful for our applications in Section \ref{sec:scalarfields}.

Hilbert spaces and other constructions from \(\mu_0\) as in Section \ref{section:notations} will be indicated by a subscript 0. The perturbations that we consider are then of the following form:
\begin{align}
\mu_\delta\left( f,g\right)&:=
\mu_0\left(f,\left(\one+\delta\right)g\right)
\,,\label{deltascalarproduct}
\end{align}
where we assume that \(\delta\) is a self-adjoint operator on \(\K_{\Rbb,0}\) whose form domain contains \(D\). To ensure that \(\mu_{\delta}\) is an inner product we require that
\(\one+\delta>0\) is strictly positive. We extend \(\delta\) by complex linearity to \(\K_0\), where we note that \(\Gamma_0\delta\Gamma_0=\delta\).

We want to ensure that the inner product \(\mu_{\delta}\) on \(D\) dominates the pre-symplectic form \(\sigma\) as in \eqref{eqn:mudominates}. This means, equivalently, that \(\one+\delta+\Sigma_0\) defines a Hermitean inner product on \(D+iD\subset\K_0\), i.e. \(\one+\delta+\Sigma_0\ge0\). A sufficient condition for this to hold is \(\delta\ge0\).

In line with our notations in Section \ref{section:notations} we let \(\K_{\Rbb,\delta}\) denote the Hilbert space completion of \(D\) in the inner product \(\mu_{\delta}\) and \(\K_{\delta}\) its complexification. There is a unitary map \(V_{\delta}:\K_{\delta}\to\K_0\) defined by \(V_{\delta}f:=\sqrt{\one+\delta}f\) for all \(f\in D\). If \(R_{\delta}\) denotes the polarisation operator on \(\K_{\delta}\), then we must have \(\sigma(f,g)=\mu_0(f,R_0g)=\mu_{\delta}(f,R_{\delta}g)\) and hence
\begin{align}
V_{\delta}R_{\delta}V_{\delta}^*&=
(\one+\delta)^{-\frac12}R_0(\one+\delta)^{-\frac12}\notag
\end{align}
and similarly
\begin{align}
V_{\delta}\Sigma_{\delta}V_{\delta}^*&=
(\one+\delta)^{-\frac12}\Sigma_0(\one+\delta)^{-\frac12}\,.\label{eqn:VSigmaV}
\end{align}
Note that \(V_{\delta}\Sigma_{\delta}V_{\delta}^*\)
is self-adjoint and has norm \(\le1\), because \(\one+\delta+\Sigma_0\ge0\). 

The perturbed state defined by \(\mu_{\delta}\) has a well-defined modular operator iff \(\one+\Sigma_{\delta}>0\), which is equivalent to \(\one+\delta+\Sigma_0>0\). (A sufficient condition for this to hold is \(\delta>0\).) Denoting the unitary map relating \(\H_{\delta}\) and \(\K_{\delta}\)
by \(U_{\delta}\), we then have
\begin{align}
V_{\delta}U_{\delta}\Delta_{\delta}U_{\delta}^*V_{\delta}^*&=\frac{\one-V_{\delta}\Sigma_{\delta}V_{\delta}^*}{\one+V_{\delta}\Sigma_{\delta}V_{\delta}^*}\notag
\end{align}
by \eqref{deltaRidentification} and hence
\begin{align}
V_{\delta}U_{\delta}K_{\delta}U_{\delta}^*V_{\delta}^*&=-\log\left(\frac{\one-V_{\delta}\Sigma_{\delta}V_{\delta}^*}{\one+V_{\delta}\Sigma_{\delta}V_{\delta}^*}
\right)
=2\artanh\left(V_{\delta}\Sigma_{\delta}V_{\delta}^*\right)\notag
\end{align}
by an identity for the inverse hyperbolic tangent.

\medskip

Now we turn to the question of quasi-equivalence of the two Gaussian states on the Weyl algebra \(\mathcal{A}\left[ D,\sigma\right] \) determined by \(\mu_0\) and \(\mu_{\delta}\). In general, we call two states on a \cstar-algebra quasi-equivalent if the GNS-representations they generate are quasi-equivalent, i.e., if every subrepresentation of the first contains a representation which is unitarily equivalent to a subrepresentation of the second \cite{haag2012local}.

Necessary and sufficient conditions for the quasi-equivalence of the two states were obtained in the general case in \cite{araki1982quasi}. The conditions consist of two parts. Firstly, the norms determined by the inner products \(\mu_0\) and \(\mu_{\delta}\) should be equivalent, which means that
\begin{align}
\one+\delta \, ,\ (\one+\delta)^{-1}&\in\mathcal{B}(\K_0)\label{condition1}
\end{align}
must both be in the space \(\mathcal{B}(\K_0)\) of bounded operators on \(\K_0\). Secondly, the difference
\begin{align}
\sqrt{\one+\delta+\Sigma_0}-\sqrt{\one+\Sigma_0}&\in\mathcal{L}^2(\K_0)\label{condition2}
\end{align}
must be in the space \(\mathcal{L}^2(\K_0)\) of Hilbert-Schmidt operators on \(\K_0\).

When the conditions \eqref{condition1} and \eqref{condition2} are satisfied, \cite{araki1982quasi} extend the quasi-free states to pure states on the Weyl algebra of a doubled symplectic space and show that these are unitarily equivalent. This construction simplifies in the case where the states $\mu_0$ and $\mu_{\delta}$ have invertible polarisation operators satisfying the additional condition \eqref{Radditionalcondition}. In this case we may identify \(H_0:=\K_{0,\Rbb}\) as a factorial standard subspace of the one-particle Hilbert space \(\H_0\). Following \cite{longo2022modular}, the doubled symplectic space can then be identified as \(H_0+H_0'\) endowed with the inner product \(\Re(,)_{\H_0}\) and with the symplectic form \(\Im(,)_{\H_0}\), which extends the symplectic form on \(H_0\). (This is called the symplectic dilation in \cite{longo2022modular}.) The completion of
this doubled symplectic space w.r.t. its inner product is \(\H_0\) as a real Hilbert space. Analogous identifications hold for the perturbed state.
Using the modular conjugations, indicated with the appropriate subscripts, there is a bijection
\begin{align}
T:H_0+H_0'&\to H_{\delta}+H_{\delta}'\notag\\
\xi+J_0\eta&\mapsto \xi+J_{\delta}\eta\notag
\end{align}
for all \(\xi,\eta\in H_0\). This bijection is symplectic, \(\Im(T(\xi_1+J_0\eta_1),T(\xi_2+J_0\eta_2))_{\H_{\delta}}=\sigma(\xi_1,\xi_2)-\sigma(\eta_1,\eta_2)=\Im(\xi_1+J_0\eta_1,\xi_2+J_0\eta_2)_{\H_0}\),
and it allows us to transport the structure of the perturbed system to the Hilbert space \(\H_0\), where we have
\begin{align}
\Re(T\xi,T\eta)_{\H_{\delta}}&=\Re(\xi,T^*T\eta)_{\H_0}\notag
\end{align}
for all \(\xi,\eta\in H_0+H_0'\). 

The results of \cite{araki1982quasi} show that the states determined by \(\mu_0\) and \(\mu_{\delta}\) are quasi-equivalent iff the states determined by \(\Re(.,.)\) and \(\Re(.,T^*T.)\) on \(H_0+H_0'\) are unitarily equivalent, which, by the results of Shale \cite{Shale1962}, is equivalent to\footnote{More precisely, Shale \cite{Shale1962} assumes the existence of a bounded Bogolyubov transformation with a bounded inverse, i.e. a bounded real linear map \(B:\H_0\to\H_0\) with a bounded inverse which preserves the symplectic form
and such that \(B^*B=T^*T\) gives the change in the inner product. If \(T\) is not bounded, then condition \eqref{condition1} of \cite{araki1982quasi} is violated, so quasi-equivalence (and hence unitary equivalence) must fail for the purified states. However, assuming \eqref{condition1} one can show that \(T\) is bounded \cite{araki1982quasi,ArakiShiraishi1971} with a bounded inverse. Furthermore, if \(T\) is bounded, \cite{araki1982quasi} prove the existence of the desired Bogolyubov transformation.}
\begin{align}
T^*T-\one&\in\mathcal{L}^2(\H_0)\,.\notag
\end{align}
Denoting the polarisation operators of the extended systems by a hat, Longo shows that this condition is equivalent to 
\begin{align}
C&:=T\hat{R}_0(T^*T-\one)=
\hat{R}_{\delta}T-T\hat{R}_0\notag
\end{align}
being in \(\mathcal{L}^2(\H_0,\H_{\delta})\) (as real Hilbert spaces), cf. \cite[Proposition 3.10]{longo2022modular}. 
Decomposing \(C\) as operators between \(H_0\), \(H_{\delta}\) and their real orthogonal complements, and exploiting the intertwining properties of \(C\) with the polarisation operators and modular conjugations, 
\cite[Theorem 4.2]{longo2022modular} finds an equivalent expression for the conditions for quasi-equivalence in terms of the polarisation operators \(R_0\) and \(R_{\delta}\), which may be re-expressed in terms of modular operators.
%
Instead of aiming for the most general setting, we will use the following more special case, which we can directly apply to our perturbed Gaussian states.
\begin{corollary}\label{Longotheorem2}
\textup{(Longo \cite[Corollary 4.4]{longo2022modular})}
Let \(H\) be a real Hilbert space w.r.t. two equivalent norms \(\mu_k\), \(k=1,2\), and let \(R_k\) be corresponding polarisation operators such that \(\mu_1(f,R_1g)=\mu_2(f,R_2g)\) for all \(f,g\in H\). Assume that \(\one+R_1^2>0\), \(\one+R_2^2>0\) and that
\(R_1\) and \(R_2\) are invertible with
\(R_1^{-1}-R_2^{-1} \in \mathcal{L}^2\!\left( H\right)\). Then, the Gaussian states determined by \(\mu_1\) and \(\mu_2\) are quasi-equivalent iff both
\begin{equation}\label{Longocondition2}
R_1^{-1}\sqrt{\one+R_1^2}-R_2^{-1}\sqrt{\one+R_2^2}  \in \mathcal{L}^2\!\left( H\right)
\end{equation}
and
\begin{equation}\label{Longocondition3}
\sqrt{\one+R_1^2}-\sqrt{\one+R_2^2}  \in \mathcal{L}^2\!\left( H\right) \, .
\end{equation}
\end{corollary}
\begin{remark}\label{Longoremark}
The invertibility of the \(R_k\) means that \(H\) is a factorial standard subspace in each of the one-particle Hilbert spaces \(\H_k\) as constructed in Section \ref{section:notations}. The square root of \(\one+R_1^2\) is the positive square root w.r.t. the inner product of \(\mu_1\) and the square root of \(\one+R_2^2\) is the positive square root w.r.t. the inner product of \(\mu_2\). If \(\mu_2(.,.)=\mu_1(.,B^2.)\) for a bounded positive operator \(B\) with bounded inverse, then \(R_2=B^{-2}R_1\) and 
\(\sqrt{\one+R_2^2}=B^{-1}\sqrt{\one+(B^{-1}R_1B^{-1})^2}B\) in terms of the positive square root w.r.t \(\mu_1\) (cf. \cite{araki1982quasi}).
Also notice that  condition~\eqref{Longocondition3} appears to be redundant. Indeed, as \(R_1\) and \(R_2\) are bounded and \(R_1^{-1}-R_2^{-1} \in \mathcal{L}^2\!\left( H\right)\), condition~\eqref{Longocondition2} already implies~\eqref{Longocondition3}. Vice versa, if we assume that, say, \(R_2^{-1}\) is bounded, then \eqref{Longocondition3}
implies \eqref{Longocondition2}, because
\begin{equation}
R_1^{-1}\sqrt{\one+R_1^2}-R_2^{-1}\sqrt{\one+R_2^2}
=\left( R_1^{-1}-R_2^{-1}\right) \sqrt{\one+R_1^2}+R_2^{-1}\left( \sqrt{\one+R_1^2}-\sqrt{\one+R_2^2}\right)
\end{equation}
will be Hilbert-Schmidt.
\end{remark}

\medskip

Now we want to supplement the above results on quasi-equivalence with some explicit estimates on the Hilbert-Schmidt norms of the operators involved. For this purpose, we will use the following known inequalities.

\begin{theorem}\label{powersstormer}
\textup{(Powers, St{\o}rmer \cite[Lemma 4.1]{powers1970free})}
Let \(A\), \(B\) be positive operators over some Hilbert space. Then,
\begin{equation}
\norm{\sqrt{A}-\sqrt{B}}_\textup{HS}^2\leq\tr\abs{A-B}\, .
\end{equation}
\end{theorem}

\begin{theorem}\label{vanHemmenAndoinequality}
\textup{(generalised van Hemmen-Ando inequality \cite[Theorem 5]{kittaneh1987inequalitiesII}, \cite[Prop. 2.1]{van1980inequality})}
Let \(A\) and \(B\) be positive, bounded operators on some Hilbert space. For any \(\alpha\ge0\) such that
\(\sqrt{A}+\sqrt{B}\geq \alpha \one\)
one has
\begin{equation}
\alpha \norm{\sqrt{A}-\sqrt{B}}_p\leq\norm{A-B}_p
\end{equation}
for any \(p\)-norm with \(p\in[1,\infty]\).
\end{theorem}

\begin{theorem}\label{AMGM}
\textup{(Operator AM-GM inequality, cf.~\cite[Theorem 1 and Appendix A]{kosaki1998arithmetic})}
Let \(A\), \(B\), \(X\) be bounded operators on some Hilbert space, with \(A\), \(B\) positive. For any unitarily invariant norm \(\norm{\,\cdot\,}\), one has
\begin{equation}
\norm{\sqrt{A}X\sqrt{B}}\leq\frac{1}{2}\norm{AX+XB}\, .
\end{equation}
\end{theorem}

Note that the arithmetic mean-geometric mean inequality of Theorem \ref{AMGM} applies in particular when \(\norm{\,\cdot\,}\) is the operator norm, the Hilbert-Schmidt norm or the trace norm.

\medskip

Condition \eqref{condition2} implies that
\begin{align}
\delta&=\frac12\sum_{\pm}(\sqrt{\one+\delta+\Sigma_0}\mp\sqrt{\one+\Sigma_0})(\sqrt{\one+\delta+\Sigma_0}\pm\sqrt{\one+\Sigma_0})
\end{align}
is also Hilbert-Schmidt. To simplify our treatment for \(\mu_{\delta}\) and \(\mu_0\) we will make the stronger assumptions \(\delta\ge0\) and
\begin{align}
\tr \delta &<\infty\,.\label{condition2'}
\end{align}
These conditions have the advantage that they do not depend on \(\Sigma\) and they will be general enough for our applications in Section \ref{sec:scalarfields}. 
They imply that both \(\one+\delta\) and \((\one+\delta)^{-1}\) are bounded and, by the Powers-Størmer inequality, Theorem \ref{powersstormer}, we have
\begin{align}
\norm{\sqrt{\one+\delta+\Sigma_0}-\sqrt{\one+\Sigma_0}}_\text{HS}^2&\le\tr\abs{(\one+\delta+\Sigma_0)-(\one+\Sigma_0)}
=\tr\abs{\delta}\,,
\end{align}
so condition \eqref{condition2} is also verified. (This observation has been made before, cf.~\cite[Appendix B]{Buchholz1974} or \cite[Section 2.1]{conti2024quasi}.) Note that our conditions on \(\delta\) may not be strong enough to have a well-defined modular operator.

To avoid confusion, we will take all square roots w.r.t. the reference inner product \(\mu_0\) and we use the notation \(R_{\delta}\) as introduced at the beginning of this section, but omitting the unitary \(V_{\delta}\), which should not lead to any confusion. We then start with the following basic estimate.
\begin{lemma}\label{lem:Restimate}
Let \(\mu_0\) be an inner product on a pre-symplectic space \((D,\sigma)\), defining a Gaussian state, and \(\mu_\delta\) a perturbation of \(\mu_0\) of the form~\eqref{deltascalarproduct}. If \(\delta\geq0\) and \(\norm{\delta}_p<\infty\) for some \(p\in[1,\infty]\), then
\begin{align}
\norm{R_{\delta}-R_0}_p&\leq \norm{\delta}_p \,.\label{additionalestimate}
\end{align}
\end{lemma}
\begin{proof}
Using \(\norm{R_0}\le1\),
\(\|(\one+\delta)^{-\frac12}\|\le1\),
\(\norm{AX}_p\le\norm{A}\norm{X}_p\) and the triangle inequality for the \(p\)-norm, we find
\begin{align}
\norm{R_{\delta}-R_0}_p&
=
\norm{((\one+\delta)^{-\frac12}-\one)R_0(\one+\delta)^{-\frac12}-R_0(\one-(\one+\delta)^{-\frac12})}_p\notag\\
&\le 2\norm{\one-(\one+\delta)^{-\frac12}}_p\notag\\
&\le 2\norm{(\one+\delta)^{-\frac12}}
\norm{\sqrt{\one+\delta}-\one}_p\notag\\
&\le\norm{\delta}_p\notag
\end{align}
by Theorem \ref{vanHemmenAndoinequality} with \(\alpha=2\).
\end{proof}
If we only assume 
\(\one+\delta>0\) instead of \(\delta\ge0\), then we can still apply Theorem \ref{vanHemmenAndoinequality} with \(\alpha=1\), showing that 
\(\norm{R_{\delta}-R_0}_p\leq 2\|(\one+\delta)^{-\frac12}\|\norm{\delta}_p\).

Sometimes we will also assume that \(R_0^{-1}\) is bounded, which is not implied by our previous assumptions, but, as we will discuss in Section \ref{sec:scalarfields}, this holds in some relevant cases.
Regarding some of the functions that appear in Longo's analysis, cf. Corollary \ref{Longotheorem2}, we then get the following results.\footnote{Strictly speaking, to compare with the results in \cite{longo2022modular} we would need to replace all functions \(X(R_{\delta})\) by \((\one+\delta)^{-\frac12}X(R_{\delta})(\one+\delta)^{\frac12}\) in order to account for the change in inner product. Note, however, that \((\one+\delta)^{\frac12}-\one\) is trace-class in our setting.}
\begin{theorem}\label{QEestimatestheorem}
Let \(\mu_0\) be an inner product on a pre-symplectic space \((D,\sigma)\), defining a Gaussian state, and \(\mu_\delta\) a perturbation of \(\mu_0\) of the form~\eqref{deltascalarproduct}. If \(\delta\geq0\) and \(\tr\delta<\infty\), then,
\begin{equation}
\norm{\sqrt{\one+R_{\delta}^2}-\sqrt{\one+R_0^2}}_\textup{HS}^2 \leq2\tr\delta \, .\label{QEcondition3estimate}
\end{equation}
If we also assume \(R_0^{-1}\) to be bounded, then
\begin{align}
\tr \abs{R_{\delta}^{-1}-R_0^{-1}} &\leq 2\norm{R_0^{-1}} \tr \delta \, ,\label{QEcondition1estimate}\\
\norm{R_{\delta}^{-1}\sqrt{\one+R_{\delta}^2}-R_0^{-1}\sqrt{\one+R_0^2}}_\textup{HS}^2 &\leq 4\,\norm{R_0^{-1}}^2\left( \tr\delta +2\left( \tr\delta\right)^2\right)\, .\label{QEcondition2estimate}
\end{align}
\end{theorem}
\begin{proof}
Let's start with condition~\eqref{Longocondition3}. By applying Theorem~\ref{powersstormer},
\begin{equation}\label{PowersStormerdelta}
\norm{\sqrt{\one+R_{\delta}^2}-\sqrt{\one+R_0^2}}_\text{HS}^2\leq\tr\abs{  \one+R_{\delta}^2-\left( \one+R_0^2\right) }=\tr\abs{ R_{\delta}^2-R_0^2}  \, .
\end{equation}
Writing
\(R_{\delta}^2-R_0^2=\frac12(R_{\delta}-R_0)(R_{\delta}+R_0)+\frac12(R_{\delta}+R_0)(R_{\delta}-R_0)\) and using \(\norm{R_{\delta}+R_0}\le2\) and the triangle inequality for the trace norm gives 
\begin{align}
\tr\abs{ R_{\delta}^2-R_0^2}&\le2\tr\abs{R_{\delta}-R_0}\,.\notag
\end{align}
Together with (\ref{PowersStormerdelta}) and Lemma \ref{lem:Restimate} for \(p=1\), this proves the estimate~\eqref{QEcondition3estimate}.

\medskip

Now assuming \(R_0^{-1}\) to be bounded to prove~\eqref{QEcondition1estimate} we have
\begin{equation}\label{QEcondition1}
	\begin{aligned}
		\tr \abs{R_{\delta}^{-1}-R_0^{-1}}&=	\tr \abs{\sqrt{\one+\delta} \ R_0^{-1}  \sqrt{\one+\delta}-R_0^{-1}}\\
		&=\tr \Big| \left( \sqrt{\one+\delta} -\one\right)  R_0^{-1}  \left( \sqrt{\one+\delta} -\one\right)\\
		&\phantom{=}+R_0^{-1}  \left( \sqrt{\one+\delta} -\one\right)+\left( \sqrt{\one+\delta} -\one\right)  R_0^{-1} \Big|\\
		&\leq \tr \abs{\left( \sqrt{\one+\delta} -\one\right)  R_0^{-1}  \left( \sqrt{\one+\delta} -\one\right)	}\\
		&\phantom{\leq}+2\norm{R_0^{-1}}\tr\abs{ \sqrt{\one+\delta} -\one}\, .
	\end{aligned}
\end{equation}
We can estimate the first term by Theorem~\ref{AMGM},
\begin{equation}\label{QEcondition1cont1}
\begin{aligned}
&\tr \abs{\left( \sqrt{\one+\delta} -\one \right)  R_0^{-1}  \left( \sqrt{\one+\delta} -\one \right)	}\\
\leq \frac{1}{2}&\tr\abs{ \left( \sqrt{\one+\delta} -\one\right)^2 R_0^{-1}+R_0^{-1} \left( \sqrt{\one+\delta} -\one \right)^2} \, ,
\end{aligned}
\end{equation}
and again from above by
\begin{equation}\label{QEcondition1cont2}
\norm{R_0^{-1}}\tr \left( \sqrt{\one+\delta} -\one \right)^2=\norm{R_0^{-1}} \tr\left( \delta -2\left(  \sqrt{\one+\delta} -\one\right) \right) \leq \norm{R_0^{-1}} \tr \delta \, ,
\end{equation}
as \(\delta\geq0\). The last term in~\eqref{QEcondition1} can be estimated by \(\norm{R_0^{-1}}\tr\delta\) as in the proof of Lemma \ref{lem:Restimate} for \(p=1\). Together with~\eqref{QEcondition1},~\eqref{QEcondition1cont1} and~\eqref{QEcondition1cont2}, this gives us~\eqref{QEcondition1estimate}.

\medskip

As noted in  Remark~\ref{Longoremark}, \eqref{QEcondition2estimate} is already implied by~\eqref{QEcondition3estimate},~\eqref{QEcondition1estimate} and the boundedness of \(R_0^{-1}\). We get
\begin{equation}\label{QEcondition2}
\hspace{-.14cm}
\begin{aligned}
&\norm{R_{\delta}^{-1}\sqrt{\one+R_{\delta}^2}-R_0^{-1}\sqrt{\one+R_0^2}}_\text{HS}^2\\ 
=\, &\Big\| \left( R_{\delta}^{-1}-R_0^{-1}\right) \sqrt{\one+R_{\delta}^2}+
R_0^{-1}\left( \sqrt{\one+R_{\delta}^2}-\sqrt{\one+R_0^2}\right) \Big\|_\text{HS}^2\\
\leq2\,&\norm{\sqrt{\one+R_{\delta}^2}}^2\norm{R_{\delta}^{-1}-R_0^{-1}}_\text{HS}^2+2\,\norm{R_0^{-1}}^2\norm{\sqrt{\one+R_{\delta}^2}-\sqrt{\one+R_0^2}}_\text{HS}^2 \, ,
\end{aligned}
\end{equation}
where \(\norm{\sqrt{\one+R_{\delta}^2}}^2\leq 1\) (since \(R_{\delta}^2\le0\)) and \(\norm{\sqrt{\one+R_{\delta}^2}-\sqrt{\one+R_0^2}}_\text{HS}^2 \leq 2 \tr \delta\) by equation~\eqref{QEcondition3estimate}. The Hilbert-Schmidt norm of \(R_{\delta}^{-1}-R_0^{-1}\) can be obtained from its trace norm~\eqref{QEcondition1estimate},
\begin{equation}\label{QEcondition2cont}
\norm{R_{\delta}^{-1}-R_0^{-1}}_\text{HS}^2\leq \left( \tr\abs{R_{\delta}^{-1}-R_0^{-1}}\,\right)^2\leq 4\,\norm{R_0^{-1}}^2 \left(  \tr \delta\right)^2 \,.
\end{equation}
Together with~\eqref{QEcondition2}  
we get~\eqref{QEcondition2estimate}.
\end{proof}
\begin{remark}
These estimates were made in terms of \(\tr\delta\), avoiding \(\norm{\delta}\), but this is not the only possible approach.
E.g., one could also estimate the left-hand side of equation~\eqref{QEcondition2cont} from above by \(2\left( 1+\|\one+\delta\|\right) \|R_0^{-1}\|^2\tr\delta\), which avoids the square of the trace of \(\delta\), but introduces its norm.
\end{remark}

It is relevant to note that up to this point we did not make use of condition~\eqref{Radditionalcondition}, or, equivalently, \(\one+\Sigma_{\delta}>0\), necessary to have a well-defined modular operator. Indeed, the class of quasi-equivalent states we have considered contains states which admit modular operators as well as states which do not. If we require the Gaussian state \(\mu_\delta\) to satisfy the condition \(\one+R_{\delta}^2>0\),  the estimate of Lemma \ref{lem:Restimate} for \(p=1\) can equivalently be written as
\begin{align}
\tr\abs{-i \tanh\left( K_{\delta}/2\right)-R_0}
&\le\tr \delta \,.\label{additionalestimate2}
\end{align}
Similarly, we can express the estimates of Theorem~\ref{QEestimatestheorem} in terms of modular Hamiltonians.
\begin{corollary}\label{cor:estimates}
Let \(\mu_0\), \(\mu_\delta\) satisfy the hypotheses of Theorem~\ref{QEestimatestheorem}, with \(\delta\geq0\), \(\tr\delta<\infty\) and \(\one+R_{\delta}^2>0\). Then, 
\begin{equation}
\norm{\frac{1}{\cosh\left( K_{\delta}/2\right)}-\sqrt{\one+R_0^2}\, }_\textup{HS}^2 \leq2\tr\delta \, .
\end{equation}
If \(R_0^{-1}\) is bounded,
\begin{align}
\tr \abs{i\coth\left( K_{\delta}/2\right)-R_0^{-1}} \quad &\leq 2\norm{R_0^{-1}} \tr \delta \, ,\\
\norm{\,i\frac{1}{\sinh\left( K_{\delta}/2\right)}-R_0^{-1}\sqrt{\one+R_0^2}\,}_\textup{HS}^2 &\leq 4\,\norm{R_0^{-1}}^2\left( \tr\delta +2\left( \tr\delta\right)^2\right)\,.
\end{align}
\end{corollary}
\begin{proof}
Note that \(\one+R_{\delta}^2>0\) iff \(\one+\Sigma_{\delta}>0\) (using \(\Gamma\Sigma\Gamma=-\Sigma\)), so we see from \eqref{deltaRidentification} that \(R_{\delta}=-i\tanh\left( K_{\delta}/2\right)\) is well-defined and we get the functions \(\sqrt{\one+R_{\delta}^2}= \left(\cosh\left( K_{\delta}/2\right) \right)^{-1}\), \(R_{\delta}^{-1} =i\coth\left( K_{\delta}/2\right) \) and \(R_{\delta}^{-1}\sqrt{\one+R_{\delta}^2}=i \left(\sinh\left( K_{\delta}/2\right)\right)^{-1}\). Then, the results follow from the estimates of Theorem~\ref{QEestimatestheorem}.
\end{proof}

\medskip

One possible extension of this discussion comes from analysing differences of more general functions of polarisation operators. 

\begin{theorem}
Under the assumptions of Theorem \ref{QEestimatestheorem}, we have
\begin{equation*}
\norm{f(\Sigma_{\delta})-f(\Sigma_0)}_\textup{HS}\leq k \, \norm{\delta}_\textup{HS} \, .
\end{equation*}
for any Lipschitz continuous function \(f\) on 
\([-1,1]\) with Lipschitz constant \(k\).
\end{theorem}
\begin{proof}
Because \(\Sigma_{\delta}\) and \(\Sigma_0\) are self-adjoint, it follows immediately from \cite[Corollary 2]{kittaneh1985lipschitz} that
\begin{equation*}
\norm{f(\Sigma_{\delta})-f(\Sigma_0)}_\textup{HS}\leq k \, \norm{\Sigma_{\delta}-\Sigma_0}_\text{HS} \, .
\end{equation*}
Lemma \ref{lem:Restimate} for \(p=2\) then yields the result.
\end{proof}

Even more generally, results of the form
\begin{equation*}
	\norm{f(A)-f(B)}_p\leq C \, \norm{A-B}_p \, ,
\end{equation*}
with \(C>0\) finite, are known to exist for Lipschitz functions for \(1<p<\infty\)~\cite{potapov2011operator,caspers2014best}.
Similarly, for \(\alpha\)-Hölder functions with \(0<\alpha<1\), one has that
\begin{equation*}
	\norm{f(A)-f(B)}_p\leq C \, \norm{\abs{A-B}^\alpha}_p
\end{equation*}
for any \(p>0\)~\cite{aleksandrov2010functions,huang2022operatorthetaholderfunctionsrespect}, with \(C>0\) finite. This allows us to include at least one of the functions of Theorem~\ref{QEestimatestheorem}, which are not Lipschitz continous, but \(f(x)=\sqrt{1-x^2}\) is \(\frac12\)-Hölder continuous. Unfortunately, the constant \(C\) in these general estimates for Lipschitz and Hölder continuous functions is in general unknown and it may not equal the Lipschitz or Hölder constant.


\section{Applications to a real linear scalar field}\label{sec:scalarfields}

In this section, we will specialise to a real linear scalar field \(\phi\) in an ultrastatic spacetime \(M=(\Rbb\times\mathcal{C},-dt^2+h)\), where \((\mathcal{C},h)\) is a complete Riemannian manifold. We consider the Klein-Gordon equation \((-\Box+m^2)\,\phi=0\), which can be written in the form \((\partial_t^2+A_m)\,\phi=0\), where \(A_m=-\Delta+m^2\) is a \(t\)-independent differential operator on \(\mathcal{C}\). For any open region \(\O\subseteq\mathcal{C}\), the symplectic space \(\left( D^{(\O)},\sigma\right) \) is given by the space of initial data
\begin{align}
D^{(\O)}&=C_0^\infty\left(\O,\Rbb\right) \oplus C_0^\infty\left( \O,\Rbb\right) \,,\label{symplecticspacefreescalar}
\end{align}
with respect to the symplectic form
\begin{align}
\sigma\left( f,g\right)&=\int_{\O} \left( f_0 g_1-f_1 g_0\right) \dd{x} \,,\label{symplecticformfreescalar}
\end{align}
where \(f=\left( f_1,f_0\right) , g=\left( g_1,g_0\right) \in D^{(\O)}\). When \(\O=\mathcal{C}\) we will write \(D\) instead of \(D^{(\O)}\).

A Gaussian state \(\omega\) is determined by a two-point distribution \(\lambda_2\) on \(M^2\), which can equivalently be characterised by its \(2\times2\)-matrix of initial data on the Cauchy surface \(x_0=0\), given by
\begin{align}
\omega_2&=
\begin{pmatrix}
\omega_{2,00}&\omega_{2,01}\\
\omega_{2,10}&\omega_{2,11} 
\end{pmatrix}
=
\begin{pmatrix}
\lambda_2\left( x,y\right)|_{x^0=y^0=0} & \partial_{y_0} \lambda_2\left( x,y\right)|_{x^0=y^0=0}\\
\partial_{x_0} \lambda_2\left( x,y\right)|_{x^0=y^0=0} & \partial_{x_0}  \partial_{y_0}\lambda_2\left( x,y\right)|_{x^0=y^0=0} 
\end{pmatrix}\,.\notag
\end{align}
The initial data are related to an inner product \(\mu\) by equation \eqref{def:omega2}.

The classical stress tensor takes the form 
\begin{align}
T_{ab}\left[ \phi\right]&=\partial_a\phi\partial_b\phi-\frac{1}{2}g_{ab}\left( \eta^{cd}\partial_c\phi\partial_d\phi+m^2 \phi^2\right) \,,\notag
\end{align}
and the energy density is the \((00)\)-component,
\begin{align}
T_{00}\left[ \phi\right]&=\frac12\left((\partial_0\phi)^2+h^{ij}\partial_i\phi\partial_j\phi+m^2 \phi^2\right) \,.\notag
\end{align}
For two Gaussian states \(\omega,\omega'\) the difference between the energy densities can be defined by a point-splitting procedure, which leads to the expression
\begin{equation}\label{energydensity}
\begin{aligned}
\epsilon_{\omega',\omega}(x)&=\frac12\Big\lbrace \left(\omega_{2,11}\left( x,y\right) -\omega'_{2,11}\left( x,y\right)\right)+\\
&\hspace{.9cm}+\Big( \sum_{i,j=1}^{n}h^{ij}(x)\partial_{x_i} \partial_{y_j}+m^2\Big) \left( \omega_{2,00}\left( x,y\right) -\omega'_{2,00}\left( x,y\right) \right)\Big\rbrace \!\Big|_{x=y}\,,
\end{aligned}
\end{equation}
when the expression in braces is a continuous function of \(x\) and \(y\) at all points on the diagonal \(x=y\) of \(\mathcal{C}\). (Here we use local coordinates \((0,x_i)\), with \((x_i)\) local coordinates on \(\mathcal{C}\).)

\subsection{Perturbations of the Minkowski vacuum}\label{ssec:globalvacuum}

We first consider the case where \((\mathcal{C},h)\) is the \(n\)-dimensional Euclidean space \(\Rbb^n\) and \(\O=\Rbb^n\). (These results should generalise to vacuum states in general ultrastatic spacetimes. See also Sec. \ref{ssec:thermal}.)

The Minkowski vacuum is determined by \(\kappa:D\to L^2(\Rbb^n):(f_1,f_0)\mapsto A_m^{-\frac14}f_1-iA_m^{\frac14}f_0\), which yields
\begin{align}
\mu_{\text{vac}}\left( f,g\right):=&\left\langle f_1, A_m^{-\frac12}\, g_1\right\rangle_{L^2\left(\Rbb^n,\Rbb\right)}+
\left\langle f_0, A_m^{\frac12}\, g_0\right\rangle_{L^2\left( \Rbb^n, \Rbb\right)}
=\left\langle f,X^2\, g\right\rangle_{L^2\left( \Rbb^n, \Rbb\right)^{\oplus 2}}
\label{def:muvac}
\end{align}
with
\begin{align}
X&=\begin{pmatrix}
A_m^{-\frac14}&0\\
0&A_m^{\frac14}
\end{pmatrix}\,.\label{eqn:defX}
\end{align}
Comparing with the general notations of Section \ref{section:notations} there is an isometry \(U_{\text{vac}}: \K_{\text{vac},\Rbb}\longrightarrow L^2\left( \Rbb^n, \Rbb\right)^{\oplus 2}\) given by
\begin{align}
U_{\text{vac}}\left(f\right)&=Xf\,,\label{unitaryvacuum}
\end{align}
which extends to a unitary \(U_{\text{vac}}: \K_{\text{vac}}\longrightarrow L^2\left( \Rbb^n\right)^{\oplus 2}\) that we denote by the same symbol. Identifying \(\K_{\text{vac}}\) with \(L^2\left( \Rbb^n\right)^{\oplus 2}\) in this way we see from \eqref{symplecticformfreescalar} that the polarisation operator is given by
\begin{align}
R_{\text{vac}}&=U_{\text{vac}}^*
\begin{pmatrix}
0&-\one\\
\one&0
\end{pmatrix}
U_{\text{vac}}\,.\label{polarizationvacuum1}
\end{align}
Note that \(R_\text{vac}^2=-\one\). In particular, condition~\eqref{Radditionalcondition}, used to define a modular operator, is not satisfied here.

We now consider a deformation of the vacuum state in the sense of Section~\ref{sec:perturbations}. The operator \(\delta\) on
\(\K_{\text{vac},\Rbb}\) is related to an operator \(\delta_{L^2}=U_{\text{vac}}\,\delta \,U_{\text{vac}}^*\) on
\(L^2\left(\Rbb^n,\Rbb\right)\), so that
\begin{align}
\mu_{\delta}\left( f,g\right)
&=\left\langle f,X \left(\one+\delta_{L^2}\right)  X\, g\right\rangle_{L^2\left( \Rbb^n, \Rbb\right)^{\oplus 2}}\label{deltascalarproductfreescalar}
\end{align}
for all \(f,g\in D\). The difference of the two-point functions of the perturbed state \(\omega_{\delta}\) and the vacuum \(\omega_{\text{vac}}\) is then given by the initial data
\begin{align}
\omega_{2,\delta}\left( f,g\right) -\omega_{2,\text{vac}}\left( f,g\right)&=
\frac12\mu_{\delta}\left( f,g\right)-\frac12\mu_{\text{vac}}\left(f,g\right) 
=\frac12\left\langle f,X \delta_{L^2} X\, g\right\rangle_{L^2} \,,\label{deltatwopointfunction}
\end{align}
with \(f,g \in D\). We can write \(\delta_{L^2}\) as a \(2\times2\)-matrix 
\(\delta_{L^2}=\begin{pmatrix} \delta_{00} & \delta_{01}\\ \delta_{10} & \delta_{11} \end{pmatrix}\) acting on
\(L^2(\Rbb^n,\Rbb)^{\oplus2}\). 
For the difference in the energy densities we find the following result, where we only consider \(m>0\) for simplicity and we treat \(A_m^{\frac14}\) as a diagonal matrix, \(\left(\begin{smallmatrix} A_m^{\frac14} & 0\\ 0 & A_m^{\frac14}\end{smallmatrix}\right)\), on \(L^2(\Rbb^n,\Rbb)^{\oplus2}\).

\begin{proposition}\label{prop:Edelta}
Assume that \(m>0\) and that \(\delta_{L^2}\) is a bounded, positive operator on \(L^2(\Rbb^n,\Rbb)^{\oplus2}\) such that
 \(A_m^{\frac14}\delta_{L^2}A_m^{\frac14}\) is a trace-class operator. Then
\begin{align}
E_{\delta}&:=\frac14\tr \left(  A_m^{\frac14}\delta_{L^2}A_m^{\frac14}\right)
\ge \frac{m}{4}\tr \delta_{L^2}\notag\,.
\end{align}
If in addition the expression in braces in equation \eqref{energydensity} for \(\omega'=\omega_{\delta}\) and \(\omega=\omega_{\textup{vac}}\) is continuous in all points on the diagonal of \(\mathcal{C}\), then 
\begin{align}
E_{\delta}&=\int_{\Rbb^n}\epsilon_{\omega_{\delta},\omega_{\textup{vac}}}\dd^nx\,.\notag
\end{align}
\end{proposition}
\begin{proof}
To obtain the inequality we estimate
\begin{align}
\tr \delta_{L^2}&=\tr \left(A_m^{-\frac14} (A_m^{\frac14}\delta_{L^2}A_m^{\frac14}) A_m^{-\frac14}\right)
\le\norm{A_m^{-\frac14}}^2\tr \left(A_m^{\frac14}\delta_{L^2}A_m^{\frac14}\right)
=4\frac{E_{\delta}}{m}\,,\label{massiveestimate}
\end{align}
using the fact that \(A_m\ge m^2\one>0\), which implies that \(A_m^{-s}\) is bounded when \(s\ge0\).

\(A_m^{-\frac12}\partial_i\) is bounded for each \(i=1,\ldots,n\), so it follows from our assumptions that the operators
\begin{align}
A_m^{\frac14}\delta_{11}A_m^{\frac14}&\notag\\
m^2A_m^{-\frac14}\delta_{00}A_m^{-\frac14}&=m^2A_m^{-\frac12}(A_m^{\frac14}\delta_{00}A_m^{\frac14})A_m^{-\frac12}
\notag\\
-\partial_iA_m^{-\frac14}\delta_{00}A_m^{-\frac14}\partial_i&=
(-\partial_iA_m^{-\frac12})(A_m^{\frac14}\delta_{00}A_m^{\frac14})(A_m^{-\frac12}\partial_i)\notag
\end{align}
are all positive and trace-class, where we view the partial derivatives as operators. The sum
\begin{eqnarray}
&\frac12\left(A_m^{\frac14}\delta_{11}A_m^{\frac14}
+m^2A_m^{-\frac14}\delta_{00}A_m^{-\frac14}
+\sum_{i=1}^n(-\partial_i)A_m^{-\frac14}\delta_{00}A_m^{-\frac14}\partial_i\right)
\end{eqnarray}
is then again trace class and by (4.9) its integral kernel is the expression in braces in equation (4.3) for the energy density
\(\epsilon_{\omega_{\delta},\omega_{\textup{vac}}}\). By our continuity assumption on this integral kernel and \cite[Corollary 3.2]{brislawn1991traceable} we then find
\begin{align}
\tilde{E}_{\delta}&:=
\int_{\Rbb^n}\epsilon_{\omega_{\delta},\omega_{\textup{vac}}}\dd^nx\notag\\
&=\frac14\tr_{L^2(\Rbb^n,\Rbb)}\! \left(A_m^{\frac14}\delta_{11}A_m^{\frac14}
+m^2A_m^{-\frac14}\delta_{00}A_m^{-\frac14}
+\sum_{i=1}^n(-\partial_i)A_m^{-\frac14}\delta_{00}A_m^{-\frac14}\partial_i\right)\,,\label{eqn:Edelta1}
\end{align}
Indeed, the results of \cite{brislawn1991traceable} allow us to weaken the continuity assumption somewhat, at the cost of additional technical complications.

To show that \(\tilde{E}_{\delta}=E_{\delta}\)
we need to use the cyclicity of the trace, taking special care, because some of the operators involved are unbounded. We compute
\begin{align}
\tr_{L^2(\Rbb^n,\Rbb)}\! \left(\sum_{i=1}^n(-\partial_i)A_m^{-\frac14}\delta_{00}A_m^{-\frac14}\partial_i\right)&=
\sum_{i=1}^n\tr_{L^2(\Rbb^n,\Rbb)}\! \left((-\partial_iA_m^{-\frac12})(A_m^{\frac14}\delta_{00}A_m^{\frac14})(A_m^{-\frac12}\partial_i)\right)\notag\\
&=\sum_{i=1}^n
\tr_{L^2(\Rbb^n,\Rbb)}\! \left((A_m^{-\frac12}(-\partial_i^2)A_m^{-\frac12})(A_m^{\frac14}\delta_{00}A_m^{\frac14})\right)\notag\\
&=\tr_{L^2(\Rbb^n,\Rbb)}\! \left((A_m^{-\frac12}(A_m-m^2)A_m^{-\frac12})(A_m^{\frac14}\delta_{00}A_m^{\frac14})\right)\notag\\
&=\tr_{L^2(\Rbb^n,\Rbb)}\! \left(A_m^{\frac14}\delta_{00}A_m^{\frac14}-m^2A_m^{-\frac14}\delta_{00}A_m^{-\frac14}\right)\,.\notag
\end{align}
Combining this identity with \eqref{eqn:Edelta1} yields 
\(\tilde{E}_{\delta}=\frac14\tr_{L^2(\Rbb^n,\Rbb)}\! \left(A_m^{\frac14}\delta_{11}A_m^{\frac14}+A_m^{\frac14}\delta_{00}A_m^{\frac14}
\right)=E_{\delta}\).
\end{proof}

Combining Proposition \ref{prop:Edelta} with the results of Section \ref{sec:perturbations} we find
\begin{theorem}\label{thm:freefieldestimates}
If \(m>0\) and \(\delta_{L^2}\) is a bounded, strictly positive operator on \(L^2(\Rbb^n,\Rbb)^{\oplus2}\) such that
\(E_{\delta}=\frac14\tr \left(  A_m^{\frac14}\delta_{L^2}A_m^{\frac14}\right)<\infty\),
then the Gaussian state \(\omega_{\delta}\) determined by \(\mu_{\delta}\) is quasi-equivalent to the Minkowski vacuum
\(\omega_{\textup{vac}}\) and we have
\begin{align}
\tr \abs{i\coth\left(\frac{K_{\delta}}{2}\right)
-R_{\textup{vac}}^{-1}}\quad&
\leq8\frac{E_\delta}{m}\, ,\label{QEestimate1freescalar}\\
\norm{\frac{1}{\cosh\left(\frac{K_{\delta}}{2}\right)}}_\textup{HS}^2\! &
\leq8\frac{E_\delta}{m}\,, \label{QEestimate3freescalar}\\  
\norm{\frac{1}{\sinh\left(\frac{K_{\delta}}{2}\right)}}_\textup{HS}^2
&\leq 16\frac{E_\delta}{m}\left( 1+8\frac{E_\delta}{m}\right)\,,\label{QEestimate2freescalar}
\end{align}
\end{theorem}
\begin{proof}
Note that \(K_{\delta}\) is well-defined, as \(\delta_{L^2}\) is strictly positive and trace-class. The proof then follows immediately from Corollary~\ref{cor:estimates} and Proposition \ref{prop:Edelta} together with the fact that \(R_{\text{vac}}^{-1}=-R_{\text{vac}}\), which implies \(\norm{R_{\text{vac}}^{-1}}=1\) and \(\sqrt{\one+R_{\text{vac}}^2}=0\).
\end{proof}

Note that a different approach to obtain inequalities involving modular operators and energies has been pursued in \cite{MuchPasseggerVerch2022}.

\subsection{Thermal states}\label{ssec:thermal}

As a further example, we consider thermal (KMS) states at inverse temperature \(\beta>0\) for a scalar field with mass \(m>0\).
We will assume that \(\mathcal{C}\) is compact, to ensure that all thermal states are quasi-equivalent
\cite{verch1994local}. Note that \(A_m\) has a pure point spectrum under this assumption (see e.g. \cite{JostGeometricAnalysis,McKeanSinger1967}).

Gaussian Thermal states are given by the inner product
\begin{align}
\mu_{\beta}\left( f,g\right)&=\left\langle f,X \coth\left( \frac{\beta}{2}A_m^{\frac12}\right)X\, g\right\rangle_{L^2\left( \mathcal{C}, \Rbb\right)^{\oplus 2}}\label{thermalscalarproduct}
\end{align}
with \(X\) as in equation \eqref{eqn:defX}. Using the identity \(\coth\left( \frac{\beta}{2} x\right)=1+\frac{2}{e^{\beta x}-1}\) for all \(x\not=0\) we can view thermal states as perturbations of the vacuum in the context of Section~\ref{sec:perturbations} with the operator \(\delta_{L^2}=\delta_\beta\) given by
\begin{align}
\delta_\beta&=\frac{2}{e^{\beta A_m^{\frac12}}-\one}\label{deltabeta}
\end{align}
as a diagonal \(2\times2\)-matrix. Notice that \(\delta_{\beta}>0\) is strictly positive (because \(A_m\) is) and trace-class, due to Weyl's asymptotic formula, cf. \cite{,JostGeometricAnalysis,McKeanSinger1967}. For similar reasons,
\(A_m^{\frac14}\delta_{\beta} A_m^{\frac14}=\frac{2A_m^{\frac12}}{e^{\beta A_m^{\frac12}}-\one}\) is also trace-class. We have
\begin{equation}\label{Sigmabeta}
\begin{aligned}
\Sigma_{\beta}&:=V_{\delta_{\beta}}\Sigma_{\delta_{\beta}}V_{\delta_{\beta}}^*\\
&=\left( \one+\delta_{\beta}\right)^{-\frac12}\Sigma_{\text{vac}}
\left( \one+\delta_{\beta}\right)^{-\frac12}\\
&=\tanh\left( \frac{\beta}{2} A_m^{\frac12}\right)\Sigma_{\text{vac}}
\end{aligned}
\end{equation}
with \(\Sigma_{\text{vac}}=\begin{pmatrix} 0 & -i\one\\ i\one & 0\end{pmatrix}\).
We will also write \(R_{\beta}=-i\Sigma_{\beta}\) and for the modular Hamiltonian we find
\begin{align}
K_{\beta}&:=V_{\delta_{\beta}}K_{\delta_{\beta}}V_{\delta_{\beta}}^*
=2\artanh\left(\Sigma_\beta\right) =\beta A_m^{\frac12}\Sigma_{\text{vac}}\label{Kbeta}
\end{align}
using \eqref{deltaRidentification} and diagonalising the \(2\times2\)-matrix \(\Sigma_{\text{vac}}\) using the unitary
\(\frac{1}{\sqrt{2}}\begin{pmatrix}-i\one&i\one\\ \one&\one\end{pmatrix}\).

We now examine the estimates of Theorem \ref{QEestimatestheorem} in this specific setting, where more explicit computations are possible.
For example, from equation~\eqref{Sigmabeta} we see that
\begin{align}
R_{\beta}^{-1}-R_{\text{vac}}^{-1}&=
i\left(\one-\coth\left( \frac{\beta}{2} A_m^{\frac12}\right)\right)\Sigma_{\text{vac}}\notag
\end{align}
and it follows that
\begin{align}
\tr\abs{R_{\beta}^{-1}-R_{\text{vac}}^{-1}} &=\tr\abs{i\left(\one-\coth\left( \frac{\beta}{2} A_m^{\frac12}\right)\right)\Sigma_{\text{vac}}}\notag\\
&=\tr \left(\coth\left( \frac{\beta}{2} A_m^{\frac12}\right)-\one\right)\notag\\
&=
\tr \delta_{\beta}
\,,\notag
\end{align}
where we recall that the trace is over \(L^2\left( \mathcal{C}, \mathbb{C}\right)\). This should be compared with the more general estimate \eqref{QEcondition1estimate}, which has an additional factor 2 on the right-hand side. Similarly, for estimate~\eqref{QEcondition3estimate},
\begin{align}
\norm{\sqrt{\one+R_{\beta}^2}-\sqrt{\one+R_{\text{vac}}^2}}_\text{HS}^2
&=\norm{\sqrt{\one+R_{\beta}^2}}_\text{HS}^2\notag\\
&=\tr\abs{\one+R_{\beta}^2}\notag\\
&=\tr\left(\tanh\left( \frac{\beta}{2} A_m^{\frac12}\right)\left(\one+\tanh\left( \frac{\beta}{2} A_m^{\frac12}\right)\right)\delta_{\beta}\right)\notag
\\
&<2\tr\delta_{\beta}\notag
\end{align}
where we used \(\delta_{\beta}=\coth\left(\frac{\beta}{2}A_m^{\frac12}\right)-\one\) and 
\(\|\tanh\left( \frac{\beta}{2} A_m^{\frac12}\right)\|\le1\). Here the inequality is strict, since
\(1-\tanh^2\left(\frac{x}{2} \right)<2\frac{2}{e^{x}-1}\), for all \(x>0\). However, as \(\beta\to\infty\) the left-hand side converges to \(2\tr\delta_{\beta}\).
Finally, regarding estimate~\eqref{QEcondition2estimate},
\begin{align}
\Big\|R_{\beta}^{-1}\sqrt{\one+R_{\beta}^2}-R_{\text{vac}}^{-1}\sqrt{\one+R_{\text{vac}}^2}\Big\|_\text{HS}^2
&=\norm{R_{\beta}^{-1}\sqrt{\one+R_{\beta}^2}}_\text{HS}^2\notag\\
&=\tr\left(-R_{\beta}^{-2}-\one\right)\notag\\
&=\tr\left(\coth^2\left(\frac{\beta}{2}A_m^{\frac12}\right)-\one\right)\notag\\
&=\tr\delta_{\beta}\left(\delta_{\beta}+2\right)\notag\\
&\le 2\tr\delta_{\beta}
+\left(\tr\delta_{\beta}\right)^2\notag
\end{align}
where we used \(\tr(A^2)\le\|A\|\tr(A)\le(\tr(A))^2\) for \(A\ge0\). This estimate has better constants than the general form \eqref{QEcondition2estimate}.

\subsection{Local perturbations of the vacuum}

To conclude this section, we will briefly comment on states on a bounded open region \(\O\subset\Rbb^n\) in an inertial time slice of Minkowski space. Here, two issues arise that make the situation noticeably more complicated than in Sections \ref{ssec:globalvacuum} and \ref{ssec:thermal}.

Due to the Reeh-Schlieder Theorem (see e.g. \cite{haag2012local}), the one-particle Hilbert space \(\H_0\) of the Minkowski vacuum remains the same, even if we restrict the theory to the bounded open region \(\O\). The same is not true, however, for the spaces \(\K_{0,\Rbb}\) and \(\K_0\). In general, we find for the region \(\O\) a subspace \(\K^{(\O)}_{0,\Rbb}\subset\K_{0,\Rbb}\), where the inner product on the symplectic space \(D^{(\O)}\) is given by the fractional operators \(A_m^{\pm\frac12}|_{\O}\), which are defined as quadratic forms on \(L^2(\O)\) essentially by restricting the integral kernels of these operators to \(\O\times\O\). It is important to note that the operators \(A_m^{\pm\frac12}|_{\O}\) are not each other's inverses.

The change in the inner product also leads to a change in the polarisation operator \(R_0\). Because \(R_0\) does not preserve the subspace \(\K^{(\O)}_{0,\Rbb}\), the desired polarisation operator \(R_0^{(\O)}\) is not the restriction of \(R_0\) to this subspace. Instead, we find a more complicated expression and, more importantly, we no longer expect to have \(R_0^2=-\one\), because the vacuum restricted to a bounded region is no longer a pure state. Indeed, we generally cannot expect that \(\|(R_0^{(\O)})^{-1}\|<\infty\), so some of our general estimates from Section \ref{sec:perturbations} no longer apply. E.g., it is well-known that a massless free scalar field on the open unit ball can be conformally mapped to a wedge region, where the Bisognano-Wichmann Theorem tells us that the modular flow is given by Lorentz boosts \cite{HislopLongo1982,BisognanoWichmann1976}. The generator of this flow has spectrum \(\Rbb\), which means that the same is true for the modular Hamiltonian of the ball \cite{LongoMorsella2023}. Hence, the spectrum of \(R_0^{(\O)}\) is \([-1,1]\) (with \(1\) and \(-1\) not in the point spectrum) and its inverse is unbounded.

This issue can be avoided if we only consider estimates that do not involve \(\|(R_0^{(\O)})^{-1}\|\). However, even in that case there are still difficulties expressing the upper bounds of our estimates in terms of the energy of a state. Indeed, the change in the inner product also renders our proof of Proposition \ref{prop:Edelta} invalid. The extension of our estimates of Theorem \ref{thm:freefieldestimates} to bounded regions will therefore require further investigation.

\bigskip

\textbf{Acknowledgements}

K.S. thanks Rainer Verch for an enlightening discussion of an approximate local modular quantum energy inequality. We also thank an anonymous referee for carefully checking the estimates and computations in our paper and for providing helpful suggestions. K.S. gratefully acknowledges the support by the Heisenberg Programme of the Deutsche Forschungsgemeinschaft (DFG) through the project "Mathematical Features of the Stress Tensor in Quantum Field Theory in Curved Spacetime" (SA 3220/1-1).

A.C. is affiliated with GNFM–INDAM (Istituto Nazionale di Alta Matematica) and acknowledges financial support under the National Recovery and Resilience Plan (NRRP), Mission 4, Component 2, Investment I.4.1, D.M. n.351, by the Italian Ministry of University and Research (MUR), funded by the European Union – NextGenerationEU.

\medskip

\textbf{Declarations}

\emph{Data availability statement:}
The authors declare that no data sets were used or created for this research.

\emph{Competing interests:} The authors have no competing interests to declare.

\printbibliography
	
\end{document}